\documentclass[11pt]{article}
\usepackage[utf8]{inputenc}
\usepackage[english]{babel}
\usepackage{amssymb,amsmath,amsfonts}
\usepackage{enumitem}
\usepackage[noend]{algpseudocode}
\usepackage[margin=1in]{geometry}
\usepackage{graphicx}

\usepackage{amsthm,algorithm,graphicx,amsmath,mdframed}
\usepackage[noend]{algpseudocode}
\newtheorem{theorem}{Theorem}[section]
\newtheorem{claim}{Claim}
\newtheorem{lemma}{Lemma}[section]

\theoremstyle{definition}
\newtheorem{definition}{Definition}[section]
\newtheorem*{remark*}{Remark}
\newtheorem{invariant}{Invariant}
\newtheorem{observation}[theorem]{Observation}
\newtheorem{question}[theorem]{Question}
\DeclareMathOperator{\polylog}{polylog}
\DeclareMathOperator{\polyloglog}{polyloglog}

\newcommand {\ignore} [1] {}

\title{Improved Dynamic Graph Coloring\footnote{A preliminary version of this paper appeared in the proceedings of ESA'18.}}

\author{Shay Solomon\thanks{IBM Research, TJ Watson Research Center, Yorktown Heights, New York, USA. Supported by the IBM Herman Goldstine Postdoctoral Fellowship.}
\and
Nicole Wein\thanks{Massachusetts Institute of Technology, Cambridge, Massachusetts, USA. Supported by an NSF Graduate Fellowship and NSF Grant CCF-1514339.}}

\date{}
\begin{document}

\maketitle

\begin{abstract}
This paper studies the fundamental problem of graph coloring in fully dynamic graphs.
Since the problem of computing an optimal coloring, or even approximating it to within $n^{1-\epsilon}$ for any $\epsilon > 0$, is NP-hard in \emph{static graphs},
there is no hope to achieve any meaningful \emph{computational} results for \emph{general graphs} in the dynamic setting.
It is therefore only natural to consider the \emph{combinatorial aspects} of dynamic coloring, or alternatively, study restricted families of graphs.

Towards understanding the combinatorial aspects of this problem, one may assume a black-box access to a static algorithm for $C$-coloring any subgraph of the dynamic graph, and investigate the trade-off between the number of colors and the number of \emph{recolorings} per update step. Optimizing the number of recolorings, sometimes referred to as the \emph{recourse} bound, is important for various practical applications.
In WADS'17, Barba et al.\ devised two complementary algorithms: For any $\beta > 0$, the first (respectively, second) maintains an $O(C \beta n^{1/\beta})$ (resp., $O(C \beta)$)-coloring while recoloring $O(\beta)$ (resp., $O(\beta n^{1/\beta})$) vertices per update.
%While these tradeoffs coincide at $d = \log n$,
%providing $O(C \log n)$-coloring with $O(\log n)$ recolorings per update,
%any slight improvement on one of these parameters triggers a significant blowup to the other.
%In particular,
%the extreme point $d = O(1)$ on the first and second tradeoff curves yields a polynomial number of recolorings and colors, respectively.
Barba et al.\ also showed that the second trade-off appears to exhibit the right behavior, at least for $\beta = O(1)$:
Any algorithm that maintains a $c$-coloring of an $n$-vertex \emph{dynamic forest} must recolor   $\Omega(n^{\frac{2}{c(c-1)}})$ vertices per update, for any constant
$c \ge 2$. Our contribution is two-fold:
\begin{itemize}
\item We devise a new algorithm for general graphs that improves significantly upon the first trade-off in a wide range of parameters: For any $\beta > 0$, we get a $\hat{O}(\frac{C}{\beta}\log^2 n)$-coloring with $O(\beta)$ recolorings per update,
where the $\hat{O}$ notation supresses $\polyloglog(n)$ factors.
In particular, for $\beta = O(1)$ we get constant recolorings with $\polylog(n)$ colors; not only is this an exponential improvement over the previous bound, but it also unveils a rather surprising phenomenon: The trade-off between the number of colors and recolorings is highly non-symmetric.
\item %We demonstrate the power of low outdegree orientations in coloring dynamic uniformly sparse graphs, such as forests (for which the WADS'17  lower bound holds):
For uniformly sparse graphs, we use low out-degree orientations to strengthen the above result by bounding the update time of the algorithm rather than the number of recolorings. Then, we further improve this result by introducing a new data structure that refines bounded out-degree edge orientations and is of independent interest. From this data structure we get a deterministic algorithm for graphs of arboricity $\alpha$ that maintains an $O(\alpha\log^2 n)$-coloring in amortized $O(1)$ time.
\end{itemize}
 \end{abstract}

% For uniformly sparse graphs, we bound the runtime rather than just the number of recolorings. In particular, for graphs of arboricity alpha which include [insert graph classes], we give an algorithm that maintains an O(\alpha \log^2 n)-coloring in amortized O(1) time. The algorithm relies on a new data structure which generalizes bounded out-degree edge orientations and may be (is?) of independent interest. 
\section{Introduction}

\subsection{Background}
Graph coloring is one of the most fundamental and well studied problems in computer science, having found countless applications over the years, ranging from scheduling and computational vision to biology and chemistry; see, e.g.~\cite{leighton1979graph, he2017graph,maharani2013digital,garey1976application,barnier2004graph, demange2009tutorial,bonchev1982chemical,khor2010application,abfalter2005nucleic,honer2013computational}, and the references therein.
A  \emph{proper $C$-coloring} of a graph $G = (V,E)$, for a positive integer $C$, assigns a color in $\{1,\ldots ,C\}$ to every vertex, so that no two adjacent vertices are assigned the same color.
The \emph{chromatic number} of the graph is the smallest integer $C$ for which a proper $C$-coloring exists.
(We shall write ``coloring'' as a shortcut for ``proper coloring'', unless otherwise specified.)

This paper studies the problem of graph coloring in fully dynamic graphs subject to edge updates.
A \emph{dynamic graph} is a graph sequence ${\mathcal G} = (G_0,G_1,\dots,G_M)$ on a fixed vertex set $V$,
where the initial graph  is $G_0 = (V,\emptyset)$ and each graph $G_i = (V,E_i)$ is obtained from the previous graph $G_{i-1}$ in the sequence by either adding or deleting a single edge.
%\footnote{[[S: Needs rewriting: Some of our results can be adapted to the fully dynamic setting with both edge and vertex updates. Since this adaptation is technically tedious, and as this setting is less standard than the edge updates setting, we omit the details from this extended abstract.]]}
We investigate general graphs as well as \emph{uniformly sparse} graphs.
% as shown in the sequel, our results and techniques reveal an interesting synergy between these inherently different graph families.
The ``uniform density'' of the graph is captured by its \emph{arboricity}: a graph $G=(V,E)$ has \emph{arboricity} $\alpha$ if $\alpha=\max_{U\subseteq V}\left\lceil\frac{|E(U)|}{|U|-1}\right\rceil$, where $E(U)=\left\{(u,v)\in E\mid u,v\in U\right\}$.
That is, the arboricity is close to the maximum \emph{density} ${|E(U)|}/{|U|}$ over all induced subgraphs
%$G[U] = (U,E(U))$
of $G$.
%While a graph of bounded arboricity is uniformly sparse, a graph of bounded density (i.e., a sparse graph) may contain a dense subgraph (e.g., on $\sqrt{m}$ of the vertices) and thus have large arboricity. Because natural to consider the graph arboricity rather than its density in domains where \emph{locality} is important; see Section \ref{s:local} for details.
%The arboricity thus reveals additional information about the graph's intrinsic density.  % is a more natural baseline for [[S: ...]]
%in other words, the arboricity parameter is much more informative than the density.//
The class of constant arboricity graphs,
%consists of all \emph{uniformly sparse} graphs. It
which contains planar graphs, bounded tree-width graphs, and in general all minor-free graphs, as well as some classes of ``real-world'' graphs,
has been subject to extensive research in the dynamic algorithms literature~\cite{BS15,BS16,PS16,solomon2018local,NS13,OSSW18, lin2012arboricity,frigioni2003fully,BF99,KKPS14,HTZ14,kowalik2007adjacency,berglin2017simple}.
A dynamic graph of \emph{arboricity} $\alpha$ is a dynamic graph such that all graphs $G_i$ have arboricity bounded by $\alpha$.
%[[S: perhaps para is too long; may want to move some defn's to preliminaries section; may want to rewrite anyway]]

It is NP-hard to approximate the chromatic number of an $n$-vertex graph to within a factor of $n^{1-\epsilon}$ for any constant $\epsilon > 0$,
let alone to compute the corresponding coloring \cite{Zuc07,KP06}.
Consequently, there is no hope to achieve any meaningful \emph{computational} results for \emph{general graphs} in the dynamic setting.
It is perhaps for that reason that the literature on dynamic graph coloring is sparse (see Section \ref{sparse}).
Nevertheless, as discussed next, one may view the area of dynamic graph algorithms as lying within the wider area of \emph{local algorithms},
in which there has been tremendous success in the context of graph coloring.

%\subsubsection{Local algorithms} \label{s:local}
When dealing with networks of large scale, it is important to devise algorithms that are intrinsically {\em local}.
Roughly speaking, a local algorithm restricts its execution to a small part of the network, yet is still able to solve a global task over the entire network.
There is a long line of work on local algorithms for graph coloring and related problems from various perspectives.
For example, seminal papers on distributed graph coloring \cite{CV86,GPS88,Linial87,ALGP89,Linial92,Luby93} laid the foundation for the area of symmetry breaking problems,
which remains the subject of ongoing intensive research.  % (see, e.g., \cite{BEK14,BEPS16,barenboim2016deterministic}).
Refer to the book of Barenboim and Elkin \cite{BE13} for a detailed account on this topic.
Additionally, graph coloring is well-studied in the areas of property testing \cite{goldreich1998property,czumaj2001testing} and local computation algorithms \cite{RTVX11,even2014deterministic}.
%\clearpage
%It is therefore only natural to consider the \emph{combinatorial aspects} of dynamic coloring, or alternatively, study restricted families of graphs.

%\subsection{Local algorithms}
%An efficient dynamic algorithm for any graph problem should be inherently \emph{local}.
%Before turning our attention to previous results in the area of dynamic graph algorithms, we believe it is

\subsubsection{Dynamic graph coloring} \label{sparse}
In light of the computational intractability of graph coloring, previous work on dynamic graph coloring is devoted mostly to heuristics and experimental results~\cite{monical2014static,preuveneers2004acodygra,yuan2017effective, hardy2017tackling,hardy2016modifying,ouerfelli2011greedy,sallinen2016graph}.
From the theoretical standpoint, it is natural to consider the \emph{combinatorial aspects} of dynamic coloring or to study restricted families of graphs;
to the best of our knowledge, the only work on this pioneering front is that of Barba et al. from WADS'17~\cite{BCKLRRV17} and Bhattacharya et al. from SODA'18~\cite{BCHN18}. Additionally, Parter, Peleg, and Solomon~\cite{ParterPS16} studied this problem in the dynamic distributed setting, and Barenboim and Maimon~\cite{BM17} studied the related problem of dynamic \emph{edge coloring}. (Our work focuses on amortized time bounds; we henceforth do not distinguish between amortized and worst-case time bounds, unless explicitly specified.)

%Towards understanding the combinatorial aspects of dynamic coloring in general graphs,
%one may assume a black-box access to a static algorithm for $C$-coloring every instance of the dynamic graph, and investigate the tradeoff between the number of colors and the number of \emph{recolorings} (i.e., the number of vertices that change their color) per update step.
%First we state the result that measures the number of recolorings, that is, the sum over all vertices $v$ of the total number of times the algorithm assigns $v$ a new color.
Barba et al. \cite{BCKLRRV17} studied the combinatorial aspects of dynamic coloring in general graphs. They assumed that 
%all graphs $G_i$ that constitute the input dynamic graph 
at all times the graph can be $C$-colored and further assumed black-box access to a static algorithm for $C$-coloring any subgraph of the current graph. They investigated the trade-off between the number of colors and the number of \emph{recolorings} (i.e., the number of vertices that change their color) per update step. The number of recolorings is an example of a \emph{recourse bound}, which counts the number of changes to the maintained graph structure done following a single update step. This measure has been well studied in the areas of dynamic and online algorithms for various fundamental problems,
such as maximal matching, MIS, approximate matching, approximate vertex and set cover, network flow and job scheduling;  see \cite{GKKV95,CDKL09,BLSZ14,BLSZ15,BHR17,GKS14,BGKPSS15,AOSS18,CHK16,GKKP17,SSTT18,bernstein2017simultaneously} and the references therein.
In some applications such as job scheduling and web hosting, a change to the underlying structure may be costly.
A low recourse bound is particularly important when the dynamic algorithm is used as a black-box subroutine inside a larger data structure or algorithm~\cite{BS16,ADKKP16}.

Barba et al. devised two complementary algorithms: for any $\beta > 0$, the first (respectively, second) maintains an $O(C \beta n^{1/\beta})$ (resp., $O(C \beta)$)-coloring while recoloring $O(\beta)$ (resp., $O(\beta n^{1/\beta})$) vertices per update step.
While these trade-offs coincide at $\beta = \log n$,
each providing $O(C \log n)$-coloring with $O(\log n)$ recolorings per update, any slight improvement on one of these parameters triggers a significant blowup to the other.
In particular, the extreme point $\beta = O(1)$ on the first and second trade-off curves yields a polynomial number of colors and recolorings, respectively.
Barba et al.\ \cite{BCKLRRV17} also showed that the second trade-off exhibits the right behavior, at least for $\beta = O(1)$:
Any algorithm that maintains a $c$-coloring of an $n$-vertex \emph{dynamic forest} must recolor $\Omega(n^{\frac{2}{c(c-1)}})$ vertices per update, for any constant
$c \ge 2$. The following question was left open.
%\begin{mdframed}
\begin{question} \label{q1}
Does the first trade-off of \cite{BCKLRRV17} exhibit the right behavior, and in particular, does a constant number of recolorings require a polynomial number of colors?
\end{question}
%\end{mdframed}

Bhattacharya et al.~\cite{BCHN18} studied the problem of dynamically coloring bounded degree graphs.
%For graphs with maximum degree bounded by $\Delta$, while maintaining a $(\Delta + 1)$-coloring in $O(\Delta)$ update time is straightforward,
%pushing the update time down to $o(\Delta)$ is highly nontrivial.
For graphs of maximum degree $\Delta$ they presented a randomized (respectively deterministic) algorithm for maintaining a $(\Delta+1)$ (resp., $\Delta(1 + o(1))$-coloring
with amortized expected $O(\log \Delta)$ (resp., $\polylog(\Delta)$) update time.
%a
%\begin{theorem} [cite soda]
%very simple variant of the randomized algorithm can maintain a $2\Delta$-coloring with expected constant update time.
%There is a fully dynamic deterministic algorithm for coloring graphs of maximum degree $\Delta$ using $(1+o(1))\Delta$ colors with amortized update time $O(\text{polylog} %\Delta)$.
%\end{theorem}
%Refer to \cite{BM17,BCHN18} for results on the related problem of \emph{edge coloring}.
These results provide meaningful bounds only when \emph{all vertices} have bounded degree.
%For the $n$-star, for example, with a maximum degree of $n-1$ but an average degree less than 2, the algorithms of \cite{BM17,BCHN18} can only provide an $n$-coloring.
%Any algorithm that maintains a $c$-coloring of an $n$-vertex \emph{dynamic forest} must recolor   $\Omega(n^{\frac{2}{c(c-1)}})$ vertices per update, for any constant
%$c \ge 2$.
%The unfortunate [[S: another adj]] conclusion is that the results of \cite{BM17,BCHN18} for bounded degree graphs cannot be extended to the wider family of bounded arboricity %graphs.
The following question naturally arises. %[[S: put q in box if possible, removed box for journal submission]]
%\begin{mdframed}
\begin{question} \label{q2}
Can we get meaningful results for the more general class of bounded arboricity graphs?
%Can the tradeoffs of \cite{BCKLRRV17} be improved for bounded arboricity graphs?
%Can we strengthen the number of recolorings bound to a similar update time bound?
\end{question}
%\end{mdframed}

Question~\ref{q2} is especially intriguging because, as shown in \cite{BCKLRRV17}, dynamic forests (which have arboricity 1) appear to provide a hard instance for dynamic graph coloring.

Parter, Peleg, and Solomon~\cite{ParterPS16} studied Question~\ref{q2} in dynamic distributed networks:
They showed that for graphs of arboricity $\alpha$ an $O(\alpha \cdot \log^* n)$-coloring can be maintained with $O(\log^* n)$ update time. % of the distributed dynamic network.
The update time in this context, however, bounds the number of \emph{communication rounds} per update,
while the number of recolorings done (and number of messages sent) per update is polynomial in $n$, even for forests.

\subsection{Our results}
We use $\hat{O}$ notation throughout to suppress $\polyloglog$ factors of $n$.

\subsubsection{General graphs}\label{sec:genres}

The following theorem summarizes our main result for general graphs.
\begin{theorem}\label{thm:intromain}
For any $n$-vertex dynamic graph that can be $C$-colored at all times, there is a fully dynamic deterministic algorithm for maintaining an $O(\frac{C}{\beta}\log^3 n)$-coloring with $O(\beta)$ (amortized) recolorings per update
step, for any $\beta>0$.
Using randomization (against an oblivious adversary), the number of colors can be reduced by a factor of $\hat{\Theta}(\log n)$
while achieving an expected bound of $O(\beta)$ recolorings.  %, except that it holds in expectation.
% while the bound on the number of recolorings holds in expectation, the bound on the number of colors is deterministic.
%There is a deterministic algorithm with the same asymptotic bound on the number of recolorings that uses $O(\frac{C}{x}\log^3 n)$ colors.
\end{theorem}
%[[S: suggest to replace $x$ by $\eta$ or some other greek letter todo]]
Theorem \ref{thm:intromain} with $\beta = O(1)$ yields $O(1)$ recolorings with $\polylog(n)$ colors,
thus answering Question \ref{q1} in the negative.
Not only is this result an exponential improvement over the previous bound of \cite{BCKLRRV17}, but it also unveils a rather surprising phenomenon: The trade-off between the number of colors and recolorings is highly non-symmetric.

We also note that the number of recolorings can be de-amortized. %Due to space constraints we defer the details to the full version. %Appendix~\ref{app:deam}).

\paragraph{A running time bound.~}
 Assuming black-box access to two efficient coloring algorithms
we can bound the running time of the algorithm from Theorem~\ref{thm:intromain}.\\
\noindent{\em Black-box static algorithm.~} Let $A_{\mathcal{G},C}$ be a static algorithm that takes as input a graph $G$ from a graph class $\mathcal{G}$ and a subset $S$ of vertices in $G$, and computes the induced graph $G[S]$ and a $C$-coloring of $G[S]$ in time $T(|S|)$.\\
\noindent{\em Black-box dynamic algorithm.~}
Let $A'$ be a fully dynamic algorithm that colors graphs of maximum degree $\Delta$ using $O(\Delta)$ colors. Such algorithms exist: there is a randomized algorithm against an oblivious adversary with $O(1)$ expected amortized update time and a deterministic algorithm with $O(\polylog(\Delta))$ amortized update time~\cite{BCHN18}. Let $T'(\Delta,n)\leq \polylog(\Delta)$ be the running time of an optimal deterministic algorithm for this problem. We state our results in terms of $T'(\Delta,n)$ to emphasize that any improvement over the deterministic algorithm of \cite{BCHN18} would yield an improvement to the running time of our algorithm.

%TODO make sure it is clear that most of the work isn't accomplished by these soda algs
%Given the algorithms $A(\mathcal{G},C)$ and $A'$ with runtimes $T$ and $T'$, respectively, we obtain:

%[[S: haven't went over this theorem, please check]]
\begin{theorem}\label{thm:introtime}
The randomized algorithm from Theorem~\ref{thm:intromain} has expected amortized update time $O\left(\frac{\beta}{n\log n}\sum_{i=0}^{\log n} 2^iT(n/2^i)\right)$ and the deterministic algorithm from Theorem~\ref{thm:intromain} has the same amortized update time with an additional additive factor of $T'(\frac{\log^2 n}{\beta},n)\leq \polylog(\frac{\log^2 n}{\beta}) $.
\end{theorem}

{\bf Remark.} %While the running time of the black-box static algorithm is prohibitively high for some graph classes (such as general graphs),
%this is not the case for the black-box dynamic algorithm, as we can take one of the aforementioned algorithms of~\cite{BCHN18}. 
%In particular, their 
The randomized black-box dynamic algorithm of ~\cite{BCHN18} that we apply in Theorem~\ref{thm:introtime} is actually a simple observation (referred to as a ``warm-up result'' in~\cite{BCHN18}) which gives a $2\Delta$-coloring with $O(1)$ expected update time. The main result of~\cite{BCHN18}, however, is an algorithm to bound the number of colors by only $\Delta+1$ (or slightly more). That is, our result does not rely on heavy machinery of prior work.

%, however the problem becomes significantly easier if the number of colors is $O(\Delta)$, which suffices for our purposes. %Specifically, we apply a simple observation (referred to as a ``warm-up result'' in~\cite{BCHN18}) to get a $2\Delta$-coloring with $O(1)$ expected update time.
%We nonetheless state our results in terms of $T'(n)$ to emphasize that any improvement over the deterministic algorithm of \cite{BCHN18} would yield an improvement in the runtime of Theorem \ref{thm:introtime}.

\subsubsection{Uniformly sparse graphs}\label{sec:resarb}

We answer Question~\ref{q2} in the positive by showing that by applying the algorithms from Theorem~\ref{thm:introtime} to arboricity $\alpha$ graphs we can obtain a bound on the update time rather than only the number of recolorings.

\begin{theorem}\label{thm:introarb}
There is a fully dynamic deterministic algorithm for graphs of arboricity $\alpha$ that for any $\beta>0$ maintains an $O((\frac{\alpha}{\beta})^2\log^4 n)$-coloring in amortized $T'(\frac{\alpha\log^3 n}{\beta^2},n)+O(\beta)\leq \polylog(\frac{\alpha\log^3 n}{\beta^2})+O(\beta)$ time per update.
Using randomization (against an oblivious adversary), the number of colors can be reduced by a factor of $\hat{\Theta}(\log n)$ and the expected amortized update time becomes
$O(\beta)$.  %, except that it holds in expectation.
% while the bound on the number of recolorings holds in expectation, the bound on the number of colors is deterministic.
%There is a deterministic algorithm with the same asymptotic bound on the number of recolorings that uses $O(\frac{C}{x}\log^3 n)$ colors.
\end{theorem}

Furthermore, we improve over this result when $\beta=o(\sqrt{\log n})$ by designing an algorithm that specifically exploits the structure of arboricity $\alpha$ graphs.

\begin{theorem}\label{thm:introarbetter}
There is a fully dynamic deterministic algorithm for graphs of arboricity $\alpha$ that maintains an $O(\alpha\log^2 n)$-coloring in amortized $\hat{O}(\polylog \alpha)$ time. Using randomization (against an oblivious adversary), the expected amortized time becomes $O(1)$.
\end{theorem}

The proof of Theorem~\ref{thm:introarbetter} relies on a new \emph{layered data structure} (LDS) for bounded arboricity graphs that we expect will be more widely applicable. 
%See Section~\ref{sec:over} for more details about the LDS %(Section~\ref{uniform})
%and the related notion of bounded out-degree edge orientations. %(Section~\ref{sec:orient})

\begin{definition} Given a dynamic graph $G$ of arboricity $\alpha$, a \emph{layered data structure (LDS)} with parameters $k$ and $\Delta$ is a partition of the vertices into $k$ layers $L_1, \dots , L_k$ so that all vertices $v$ have at most $\Delta$ neighbors in layers equal to or higher than the layer containing $v$.
\end{definition}

\begin{theorem}\label{thm:intro-lds}
Let $A''$ be an algorithm for arboricity $\alpha$ graphs that maintains an orientation of the edges with out-degree at most $D$ that performs amortized $F(n)$ flips per update. Then there is an algorithm to maintain an LDS along with the graph induced by each layer, for a fully dynamic graph of arboricity $\alpha$ with $k=O(\log n)$ and $\Delta=O(D+\alpha\log n)$ in amortized deterministic time $O(F(n))$.
\end{theorem}

\subsection{Technical overview}\label{sec:over}
\subsubsection{Low out-degree dynamic edge orientations}\label{sec:orient}
%\noindent
%\vspace{6pt}
%\\
%There is a rapidly growing body of work on

All of our results are, in different ways, intimately related to the dynamic edge orientation problem for arboricity $\alpha$ graphs, where the goal is to dynamically maintain a low out-degree orientation of the edges in a graph (an orientation with out-degree $\alpha$ always exists~\cite{nash1964decomposition}).
%As a starting point of our algorithm for general graphs, we discover and exploit a connection between the dynamic edge orientation problem and dynamic coloring
Our algorithm for general graphs (outlined in Section~\ref{sec:overgen}) is inspired by an algorithm for the dynamic edge orientation problem. Our algorithm for bounded arboricity graphs from Theorem~\ref{thm:introarb} uses a dynamic edge orientation algorithm as a black-box. Our algorithm for bounded arboricity graphs from Theorem~\ref{thm:introarbetter} uses a dynamic edge orientation to define a potential function useful in the running time analysis (outlined in Section~\ref{sec:ldsover}).

%The basic goal in dynamic low outdegree orientations (shortly, DLOO) is to maintain an edge orientation in which the maximum outdegree never exceeds the graph arboricity (which is an immediate lower bound for the outdegree)  by too much while keeping a tab on the update time
%[all the citations you have].
%As a starting point of our dynamic coloring algorithm for general graphs, we discover and exploit a connection between DLOO
%(in bounded arboricity graphs) and dynamic coloring in general graphs.
%As mentioned in Section ??, a crucial ingredient of our dynamic algorithm is a black-box static coloring algorithm.
%Our dynamic algorithm will periodically run this black-box on carefully computed induced subgraphs $G[V']$, where $V' \subseteq V$.
%As a result of any such static computation, all the vertices in subsets $V'$ may get re-colored, but the number of re-colorings depends linearly on $V'$ and not on the size %(number of edges) of the subgraph induced by $V'$.
%For low outdegree orientations, on the other hand, the number of edge re-oreintations triggered by a static computation depends linearly on the size of the subgraph induced by %$V'$,
%hence to perform effciently it is crucial to restrict the attention to low arboricity graphs.

Brodal and
Fagerberg~\cite{BF99} initiated the study of the dynamic edge orientation
problem and gave an algorithm that maintains an $O(\alpha)$ out-degree
orientation in amortized $O(\alpha+\log n)$ time. To analyze this algorithm, they reduced the ``online'' setting, where we have no knowledge of the future,
to the ``offline'' settings, where we know the entire sequence of edge updates in advance. Thus, in the the subsequent results, it sufficed to consider only the offline setting. Kowalik~\cite{kowalik2007adjacency} used an elegant argument to derive a result complementary to~\cite{BF99}: one can maintain an $O(\alpha\log n)$ out-degree orientation in amortized
$O(1)$ time. He, Tang, and Zeh~\cite{HTZ14} completed the picture with a
trade-off bound: for all $\beta\geq 1$, one can maintain an
$O(\beta\alpha)$ out-degree orientation in amortized
$O(\frac{\log n}{\beta})$ time. The worst-case update time of this problem has also been studied
by Kopelowitz et al.~\cite{KKPS14} and
Berglin and Brodal~\cite{berglin2017simple}.

Dynamic bounded out-degree orientations are a key ingredient in a number of dynamic algorithms for graphs of bounded arboricity~\cite{kowalik2006oracles,kowalik2004fast,BS15,BS16,NS13,OSSW18, frigioni2003fully,dvovrak2013dynamic}, as well as in dynamic algorithms for general graphs~\cite{Sol16,BS15,BS16,bodwin2016fully}.

\subsubsection{Overview of algorithm for general graphs}\label{sec:overgen}

We apply two black-box coloring algorithms defined in Section~\ref{sec:genres}, one static and one dynamic. For each vertex $v$, if it is assigned color $c_1$ by the static algorithm and color $c_2$ by the dynamic algorithm, its true color is defined by the pair $(c_1,c_2)$.

Periodically, we run the static algorithm using a carefully chosen subset of vertices as input. To select these subsets, we keep track of the \emph{recent degree} of each vertex $v$: the number of edges incident to $v$ that were inserted since the last time $v$ was included as input to an instance of the static algorithm. Then, we choose the vertices of highest recent degree as input to the static algorithm, thus setting the recent degree of these vertices to zero. By repeatedly setting the recent degree of the highest recent degree vertices to zero, we obtain a bound on the maximum recent degree in the graph. Then we apply the dynamic algorithm for bounded degree graphs on only the edges that contribute to recent degrees.

We can further reduce the maximum recent degree in the graph by employing randomization: In addition to the vertices already chosen to participate in the static algorithm, we randomly select some vertices incident to newly inserted edges.
%In our randomized algorithm, we assume an oblivious adversary. That is, the sequence of edge updates cannot depend on the colors of the vertices. [todo say something like oblivious adversary assumption is standard?]

To obtain an upper bound on the maximum recent degree at all times, we model the changes in recent degree by an online 2-player balls and bins game. The game was first introduced in the late 80s~\cite{LO88,dietz1987two} and has found a number of applications in the dynamic algorithms literature for obtaining worst-case guarantees~\cite{dietz1991persistence,charikar2017fully,ACK17,Thorup05, AT07,mortensen2006fully,Wulff-N17,berglin2017simple,dietz1994constant,kaporis2005isb}. To the best of our knowledge, our techniques are the first to demonstrate improved amortized guarantees using the game.
We anticipate that this game will find additional applications in amortized algorithms as well as in translating offline strategies to online strategies.
%[Cite paper] analyzes a number of variants of the game, including a randomized variant, which we use to analyze our randomized algorithm.

The main technical content that remains are the details of each instance of the static algorithm: we have not specified when to run each instance, the precise subset of vertices to input, and which palette of colors to draw from. Understanding these details illuminates the key insight that allows us to improve the number of colors from the polynomial bound in~\cite{BCKLRRV17} to polylogarithmic. We hierarchically bipartition the update sequence into $\log_2 n$ levels of nested time intervals and at the end of each interval, we apply the static algorithm. We use a separate palette of colors for each level of intervals but for all instances of the static algorithm on the same level we use the same palette. Consequently, we need to ensure that vertices colored at the end of different intervals on the same level do not have conflicting colors. To do this, we ensure the structure of the intervals is such that if we color a vertex $v$ at the end of an interval on some level $L$, then before the end of the next interval on level $L$, $v$ has been recolored due to the end of an interval on a different level. 

%To avoid color conflicts caused by this reuse of colors, we ensure that when a vertex participates in an instance of the static algorithm at the end of an interval $I$ on a level $L$, it also participates in the instance of the static algorithm at the end of every superinterval of $I$; thus it is recolored once for each superinterval. This recoloring frees the level $L$ color palette for future instances of the static algorithm at level $L$. In summary, the hierarchical partition of the update sequence into levels provides the structure that allows us to reuse colors without creating color conflicts.

This partition of the update sequence is inspired by the offline algorithm of~\cite{HTZ14} for the dynamic edge orientation problem. Adapting their ideas to our setting requires overcoming two main hurdles: a) transitioning from graphs of bounded arboricity to general graphs, and b) transitioning from the offline setting to the online setting.
%In summary, the hierarchical partition of the update sequence into levels provides the structure that allows us to reuse colors in instances of the static algorithm without creating color conflicts.

%A similarity between this approach and the approach of [cite WADS paper] is that both define a partition of the vertices into levels which each use a disjoint color palette. [WADS paper] uses a polynomial number of levels, while we get away with only a logarithmic number of carefully structured levels. This hierarchical partition approach is inspired by He, Tang, and Zeh's [cite] approach to the dynamic edge orientation problem on bounded arboricity graphs [I'm assuming we've introduced this problem already]. [Can we say that the alg and analysis are simple as another selling point?] Add more selling points here like online vs offline.

\subsubsection{Overview of algorithm for low arboricity graphs} \label{uniform}
The proof of Theorem~\ref{thm:introarb} is based on the following observation: the black-box static algorithm used in Theorem~\ref{thm:intromain} can be made efficient if $\mathcal{G}$ is the class of arboricity $\alpha$ graphs and we have access to a low out-degree orientation of the graph.
%If particular, if $G\in \mathcal{G}$, then if we have an orientation of the edges of $G$ with maximum out-degree $d$, for any subset $S$ of the vertices in $G$ there is a very simple $O(nd)$ time algorithm to compute $G[S]$ and a $2\alpha$-coloring of $G[S]$. Furthermore, there are efficient dynamic algorithms for maintaining a bounded out-degree orientation of the edges of $G$. [describe the algs here or somewhere else?] These observations together imply the following result.

%A standard technique for designing dynamic algorithms for graphs of low arboricity is to maintain a low out-degree orientation of the edges. Maintaining such an orientation allows each vertex to have complete information about its neighborhood in the following way. When a vertex changes its status, it notifies only its few out-neighbors. Then, when a vertex needs information about its neighborhood, it already has information from all of its in-neighbors and needs to fetch information only from its few out-neighbors. [give examples and citations of which problems they have been useful for]

The proof of Theorem~\ref{thm:introarbetter} concerns the LDS (defined in Section~\ref{sec:resarb}). The definition of the LDS is inspired by the following property of arboricity $\alpha$ graphs: there exists an ordering of the vertices $v_1,\dots,v_n$ such that every vertex has at most $2\alpha$ neighbors that appear after it in the ordering~\cite{arikati1997efficient}. Given such an ordering, consider the procedure of iteratively removing the vertices from the graph in order (or adding the vertices to the graph in reverse order) so that when each vertex is removed (or added) its degree with respect to the current graph is only $2\alpha$. This procedure has been a key ingredient in algorithms in a variety of settings including distributed algorithms~\cite{barenboim2010sublogarithmic}, parallel algorithms~\cite{arikati1997efficient}, property testing~\cite{eden2018testing}, and social network analysis~\cite{jain2017fast,seidman1983network,danisch2018listing}. We are the first to devise a data structure that dynamically maintains (an approximate version of) this ordering.

%However, low out-degree orientations do not seem to suffice for solving certain dynamic problems such as coloring. When an edge is inserted between two vertices of the same color such that, it is unclear how to efficiently recolor one of the vertices $v$, even with complete knowledge of its neighborhood. In particular, the vertices in $v$'s neighborhood could cover all of the colors in the entire palette, in which case coloring $v$ any color $c$ would necessitate the recoloring of all of the vertices in $v$'s neighborhood that are colored $c$.

%One shortfall of a low out-degree orientation is that it is an inherently \emph{local} data structure; each vertex only keeps track of information about its immediate neighborhood. We introduce a generalization of the edge orientation data structure that maintains \emph{global} information about the graph. In particular, we utilize the following property of arboricity $\alpha$ graphs: there exists a partition of the vertices into $O(\log n)$ layers such that every vertex has degree $O(\alpha)$ to vertices in equal or higher layers. Our \emph{layered data structure} dynamically maintains an approximate version of such a partition into layers. In particular, we allow a logarithmic slack factor in the degree to equal or higher layers. We anticipate that the data structure could be useful for other problems.

The LDS is useful for maintaining a proper coloring of a graph because the graph induced by each layer of vertices has low degree. Thus, we can apply a dynamic algorithm for graphs of bounded maximum degree on the graph induced by each individual layer. Then, because there are not too many layers in total, we can use a disjoint palette of colors for each layer.

On the other hand, simply using a low out-degree orientation of the edges does not seem to suffice for solving dynamic coloring. In general, one shortfall of a low out-degree orientation is that it is an inherently \emph{local} data structure; each vertex only keeps track of information about its immediate neighborhood. In contrast, the LDS maintains a $\emph{global}$ partition of the vertices into layers. Furthermore, the LDS is designed to store strictly more information than a bounded out-degree edge orientation; by orienting all edges in the LDS from lower to higher layers, we get a bounded out-degree edge orientation. We anticipate that the LDS could be useful for solving more dynamic problems for which a bounded out-degree edge orientation does not appear to suffice.

%[maybe this is sufficient or maybe we should state the LDS result in the results section and add something here about how we prove the LDS result, using edge orientations as a potential function, what do you think?] 
\section{Algorithm for general graphs}\label{sec:gen}
In this section we prove Theorems~\ref{thm:intromain}, \ref{thm:introtime}, and \ref{thm:introarb}.
% to the full version.%Appendix~\ref{app:runtime}.

\begin{theorem}[Restatement of Theorem~\ref{thm:intromain}]\label{thm:body-main}
There is a fully dynamic deterministic algorithm for maintaining an $O(\frac{C}{\beta}\log^3 n)$-coloring with $O(\beta)$ (amortized) recolorings per update
step, for any $\beta>0$. 
Using randomization (against an oblivious adversary), the number of colors can be reduced to $O(\frac{C}{\beta}\log^2 n(\log\log n+\log \beta))$
while achieving an expected bound of $O(\beta)$ recolorings.  
\end{theorem}

The algorithm is as follows. Periodically, we run the black-box static algorithm on a subset of vertices, to be specified later. At all times, each vertex $v$ is assigned a color $c_1$ by the black-box static algorithm (from the last time $v$ was input to an instance of the static algorithm) and a color $c_2$ by the black-box dynamic algorithm. The true color of $v$ is defined by the pair $(c_1,c_2)$, so the total number of colors is the product of the number of colors used in each black-box algorithm. 
%As mentioned in the algorithm overview (Section~\ref{sec:overgen}), 
To specify the subsets of vertices taken as input to the static algorithm, we define a hierarchical partition of the update sequence. First, we describe this partition, then we describe how to apply the static algorithm, and then we describe how to apply the dynamic algorithm.

\subsection{Partition of update sequence}\label{sec:part}

We partition the update sequence (without knowing its contents) into a set of intervals as follows. An interval is said to be of \emph{length} $\ell$ if it contains $\ell$ update steps. We partition the entire update sequence into intervals of length $n\ell$ for some parameter $\ell$ (which we will later set to $\frac{\log n}{\beta}$). We say that this set of intervals is on \emph{level} 0. Next, for each $i = 1,\ldots,\log_2 n$, the level-$i$ intervals are obtained from the $i-1$-level intervals by splitting each $i-1$ interval in two subintervals of equal length. Note that the intervals on level $\log n$ are of length $\ell$ and in general the intervals on level $i$ are of length $n\ell/2^i$.

It will be easier to work with these intervals if no two have the same ending point. So, for every set of intervals with the same endpoint, we remove all intervals except for the one with the lowest numbered level. The resulting set of intervals, shown in Figure~\ref{fig:intervals} is the set of intervals that we work with in the algorithm.\\

\begin{figure}[h]
  \centering
    \includegraphics[width=0.5\textwidth]{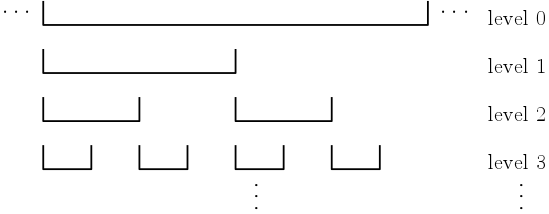}
     \caption{The set of interals.}
     \label{fig:intervals}
\end{figure}

\subsection{Applying the black-box static algorithm}

At the end of each interval, we apply the black-box static algorithm. For each interval $I$, let $\mathcal{A}_I$ be the instance of the black-box algorithm that is executed at the end of interval $I$. If $I$ is an interval on level $i$, we say that $\mathcal{A}_I$ is on level $i$. For each level, we use a separate palette of $C$ colors, and all instances of the algorithm on the same level use the same palette of colors. In particular, if $\mathcal{A}_I$ is on level $i$, it uses the $C$ colors in the range from $i\cdot C+1$ to $(i+1)C$. 

We determine the input $S$ to each $\mathcal{A}_I$ as follows. If $I$ is on level 0, the input $S$ is simply the entire vertex set. Otherwise, we decide the input based on the update sequence. For each vertex $v$, we keep track of its \emph{recent degree}, defined as the number of edges incident to $v$ that were inserted since the last time $v$ was included as input to an instance of the static algorithm. For each interval $I$, we let $v_I$ be the vertex of highest recent degree at the end of interval $I$ (breaking ties arbitrarily). For the deterministic algorithm, the input $S$ to each $\mathcal{A}_I$ is the set of vertices $\{v_{I'}|I' \text{ is a subinterval of } I\}$ (where an interval is considered a subinterval of itself). 

For the randomized algorithm, in addition to $v_I$ we select another vertex $u_I$ at the end of each interval $I$. Specifically, we pick uniformly at random an edge insertion $(y,z)$ from the last $\ell$ updates (if one exists) and then we let $u_I$ be either $y$ or $z$, chosen at random. Then the input $S$ to each $\mathcal{A}_I$ is the set of vertices $\{v_{I'},u_{I'}|I' \text{ is a subinterval of } I\}$. 

We note that each interval on level $\log_2 n$ contains only 1 subinterval (itself), and generally, each interval on level $i$ contains $n/2^i$ subintervals. Thus, each $\mathcal{A}_I$ on level $i$ takes $O(n/2^i)$ vertices as input. 

\subsection{Applying the black-box dynamic algorithm}

We apply the black-box dynamic algorithm on the graph with the full vertex set but only the edges that count towards the recent degree of both of its endpoints. Specifically, if $G$ denotes the input dynamic graph then the dynamic graph $G'$ that we input to the black-box dynamic algorithm is defined as follows. $G'$ is initially the empty graph on the same vertex set as $G$ and whenever there is an edge update to $G$, the same edge is updated in $G'$. Additionally, when a vertex $v$ is included as input to the static algorithm, every edge incident to $v$ is deleted from $G'$. 

To apply the black-box dynamic algorithm, we need to show that $G'$ has bounded maximum degree. To do this, we apply an online 2-player balls and bins game. The game begins with $N$ empty bins. The goal of Player 1 is to maximize the size of the largest bin and the goal of Player 2 is the opposite. At each step, the players each make a move according the following rules.
\begin{itemize}
\item Player 1 distributes at most $k$ new balls to its choice of bins.
\item Player 2 removes all of the balls from the largest bin (breaking ties arbitrarily).
\end{itemize}

\begin{theorem}[\cite{dietz1987two}]\label{thm:bins1}
In the balls and bins game, every bin always contains $O(k\log N)$ balls.
\end{theorem}

A randomized variant of the game will be useful in analyzing our randomized algorithm. In this variant, in addition to emptying the largest bin, Player 2 also chooses a number $i$ from $[k]$ uniformly at random and empties the bin to which Player 1 added its $i^{th}$ ball during its last turn. Player 1 is oblivious to the behavior of Player 2.

\begin{theorem}[\cite{dietz1991persistence}]\label{thm:bins2}
In the randomized variant of the balls and bins game, in a game with $N$ moves every bin always contains $O(k\log\log N+k\log k)$ balls with high probability.\footnote{``High probability'' means that for all $c>0$, there is an $N$ such that the probability is at least $1-N^{-c}$}
\end{theorem}

Recall that $\ell$ is a parameter introduced in Section~\ref{sec:part}.

\begin{lemma} \label{lem:deg}
In the deterministic algorithm the maximum degree of $G'$ is always $O(\ell\log n)$. In the randomized algorithm the maximum degree of $G'$ is always $O(\ell\log\log n+\ell\log \ell)$.
\end{lemma}

\begin{proof}
We will argue that in the balls and bins game with $N=n$ and $k=2\ell$, the number of balls in the largest bin is an upper bound for the maximum degree of $G'$. Then, applying Theorems~\ref{thm:bins1} and~\ref{thm:bins2} completes the proof.

We first note that by construction, the degree of each vertex $v$ in $G'$ is at most the recent degree of $v$ so it suffices to bound recent degree. (In particular, the recent degree of $v$ could be larger because it counts edges to vertices that have recently been included as input to the static algorithm.)

The only way for the recent degree of a vertex $v$ to increase is due to the insertion of an edge incident to $v$. On the other hand, the recent degree of a vertex $v$ decreases when a) an edge incident to $v$ is deleted causing its recent degree to decrement, and b) $v$ is included as input to the static algorithm causing its recent degree to be set to 0. 

We consider the special case of the balls and bins game where for each edge insertion $(u,v)$, Player 1 places one ball in the bin corresponding to $u$ and one ball in the bin corresponding to $v$. Then, when each interval ends (which happens once every $\ell$ updates), Player 2 moves. Recall that at this point the recent degree of $v_I$ is set to 0 (and in the randomized algorithm, so is that of $u_I$). It is clear from this description that the deterministic and randomized balls and bins games parallel all of the increases and some of the decreases in recent degree in our deterministic and randomized algorithms, respectively. From here, it is easy to verify that the number of balls in the largest bin is an upper bound for the maximum recent degree in both the deterministic and randomized settings. For the sake of completeness, we prove this fact formally in the appendix. %Appendix~\ref{app:bins}.
\end{proof}

\subsection{Correctness}\label{sec:correct}

We will show that our algorithm produces a proper coloring after every update. Recall that the color of each vertex $v$ is defined by the pair of colors $(c_1,c_2)$ where $c_1$ is the color assigned to $v$ by the black-box static algorithm and $c_2$ is the color assigned to $v$ by the black-box dynamic algorithm. 

Consider an edge $(u,v)$ in the graph at a fixed point in time. We will show that our algorithm assigns different colors to $u$ and $v$. If $(u,v)$ is included in the input to the black-box dynamic algorithm (i.e. if $(u,v)$ is in $G'$), then its two endpoints are assigned different colors by this algorithm, and are thus assigned different colors by the overall algorithm. 

Otherwise, by the definition of the input to the black-box dynamic algorithm, after the edge $(u,v)$ was last inserted at least one of $u$ or $v$ was included as input to the static algorithm. We claim that $u$ and $v$ are assigned different colors by the static algorithm. If $u$ and $v$ were last colored by the same instance $\mathcal{A}_I$ of the static algorithm, then $\mathcal{A}_I$ was executed after the edge $(u,v)$ was inserted (by assumption). Thus, the edge $(u,v)$ was included as input to $\mathcal{A}_I$, causing $u$ and $v$ to be assigned different colors. If $u$ and $v$ were last colored by instances of the static algorithm on different levels, then they are assigned different colors since each level uses a separate palette of colors. 

The only remaining case is that $u$ and $v$ were last included as input to the static algorithm by two different instances of the static algorithm on the same level $i$. We will show that this is impossible. This case is the crux of the correctness argument and the reason that we define the intervals in precisely the way that we do. It cannot be the case that $i=0$ since every vertex is recolored at the end of every interval on level 0. Suppose by way of contradiction that $u$ was most recently colored by $\mathcal{A}_{I}$ (the instance of the static algorithm at the end of interval $I$) and $v$ was most recently colored by $\mathcal{A}_{I'}$ where interval $I$ comes before interval $I'$ and both are on level $i$. We will show that between the end of interval $I$ and the end of interval $I'$, $u$ is recolored by an instance of the static algorithm on a level $j<i$ (a contradiction). By the construction of the intervals (see Figure~\ref{fig:intervals}), between the ending points of $I$ and $I'$ is the end of an interval $I''$ on a level $j<i$ that contains interval $I$ as a subinterval. By the definition of the algorithm, every vertex that is included as input to $\mathcal{A}_I$ is also included as input to $\mathcal{A}_{I''}$. Thus, $u$ is recolored on level $j$ before $\mathcal{A}_{I'}$ was executed, a contradiction.

\subsection{Analysis}

\subsubsection{Static algorithm}
\paragraph{Number of colors} The static algorithm uses $C$ colors per level and there are $O(\log n)$ levels, for a total of $O(C\log n)$ colors.

\paragraph{Number of recolorings}  In the deterministic algorithm, each interval $I$ has an associated vertex $v_I$ and in the randomized algorithm, each interval has two associated vertices $v_I$ and $u_I$. Each such vertex is included as input to the static algorithm for all superintervals of $I$. Since there are $O(\log n)$ levels and each level consists of a set of disjoint intervals, each interval has at most $O(\log n)$ superintervals. Thus, for each interval $I$, $v_I$ and $u_I$ are included as input to $O(\log n)$ instances of the static algorithm. Every interval ends after a multiple of $\ell$ updates so the number of recolorings is amortized $O(\frac{\log n}{\ell})$.

\subsubsection{Dynamic algorithm}
\paragraph{Number of colors} Given a dynamic graph of maximum degree $\Delta$, the black-box dynamic algorithm maintains an $O(\Delta)$-coloring. By Lemma~\ref{lem:deg}, $G'$ (the graph input to the black-box dynamic algorithm) has maximum degree $O(\ell\log n)$ in the deterministic setting and $O(\ell\log\log n+\ell\log \ell)$ in the randomized setting. The randomized bound is with high probability and in the low probability event that the maximum degree exceeds the bound, we will immediately end all intervals, thereby recoloring the entire graph. Thus, the runtime bound is probabilistic but the bound on the number of colors is not.\\
\noindent{\emph{Number of recolorings.}~} Using the following simple greedy algorithm as our black-box dynamic algorithm, we get a single recoloring per update. When an edge is added between two vertices of the same color, simply scan the neighborhood of one of them and recolor it with a non-conflicting color. If the maximum degree of the graph is $\Delta$, this algorithm produces a $\Delta+1$ coloring.

\subsubsection{Combining the static and dynamic algorithms}
\paragraph{Number of colors} Recall that if a vertex $v$ is assigned color $c_1$ by the black-box static algorithm and color $c_2$ by the black-box dynamic algorithm, then our algorithm assigns $v$ the color $(c_1,c_2)$. So the number of colors is the product of the number of colors used in each black-box algorithm, which is $O(C\ell\log^2 n)$ for the deterministic algorithm and $O(C\ell\log n(\log\log n\log \ell))$ for the randomized algorithm.\\
\noindent{\emph{Number of recolorings.}~} The total number of recolorings is the sum of the number of recolorings performed in each of the black-box algorithms, which is $O(\frac{\log n}{\ell})$.

Setting $\ell=\frac{\log n}{\beta}$ completes the proof.

\subsection{Time bound}\label{app:runtime}
\begin{theorem}[Restatement of Theorem~\ref{thm:introtime}]\label{thm:bodytime}
The randomized algorithm from Theorem~\ref{thm:intromain} has expected amortized update time $O\left(\frac{\beta}{n\log n}\sum_{i=0}^{\log n} 2^iT(n/2^i)\right)$ and the deterministic algorithm from Theorem~\ref{thm:intromain} has the same amortized update time with an additional additive factor of $T'(\frac{\log^2 n}{\beta},n)$.
\end{theorem}
\begin{proof}

~\subsubsection{Static algorithm}
We defined $v_I$ as the vertex of highest recent degree at the end of each interval $I$. Although there are data structures to find $v_I$ in constant time, it suffices for the analysis of this algorithm to spend $O(\log n)$ time to find each $v_I$. An interval ends once every $\ell$ updates so the amortized time is $\frac{\log n}{\ell}$.

At the end of every interval in the randomized algorithm, $u_I$ is chosen randomly from a distribution over only the most recent $\ell$ updates, so this takes constant amortized time. Then, after finding $v_I$ and $u_I$, determining the input to each $\mathcal{A}_I$ takes time linear in the size of the input to $\mathcal{A}_I$. 

We now analyze the time to run all of the instances of the static algorithm. Recall that given a graph $G$ and a subset $S$ of the vertices, the static algorithm computes $G[S]$ and a $C$-coloring of $G[S]$ in time $T(|S|)$. Consider a level 0 interval $I$. At the end of interval $I$, we run the algorithm on the entire vertex set, which takes time $T(n)$.  By construction of the intervals, each interval on level $i$ contains $n/2^i$ subintervals. Thus, the static algorithm at the end of each interval on level $i$ takes $n/2^i$ vertices as input so each such static algorithm runs in time $T(n/2^i)$. For all levels $i$, there are $2^i$ subintervals of $I$ on level $i$. Thus, it takes total time $2^iT(n/2^i)$ to run all instances of the static algorithm on level $i$ that are executed during interval $I$. Therefore, the total time to run all instances of the static algorithm that are executed during interval $I$ (including the one at the end of interval $I$) is $\sum_{i=0}^{\log n} 2^iT(n/2^i)$. Interval $I$ is of length $n\ell$ so the amortized time is $\frac{1}{n\ell}\sum_{i=0}^{\log n} 2^iT(n/2^i)$.

\subsubsection{Dynamic algorithm}
We note that the number of updates to $G'$ is at most twice the number of updates to $G$ since for each edge $(u,v)$ inserted to $G$, $(u,v)$ is inserted to and deleted from $G'$ at most once. Thus, maintaining $G'$ takes constant amortized time.

Recall that the black-box dynamic algorithm has amortized expected update time $O(1)$ in the randomized setting and amortized update time $T'(\Delta,n)\leq O(\polylog(\Delta))$ in the deterministic setting.

In the randomized algorithm, when the maximum degree of $G'$ exceeds the stated bound, we immediately end all intervals, thereby recoloring the entire graph. For large enough $n$, this happens with probability less than $1/n^2$. Each time this happens, we pay an extra $T(n)$ time to recolor the entire graph. Thus, this takes amortized time $T(n)/n^2$ in expectation.

\subsubsection{Combining the static and dynamic algorithms}
The total amortized update time is the sum of the amortized running times of each of the black-box algorithms, which is $\frac{1}{n\ell}\sum_{i=0}^{\log n} 2^iT(n/2^i)$ in expectation for the randomized algorithm, and with an additional additive factor of $T'(\frac{\log^2 n}{\beta},n)$ for the deterministic algorithm. (This expression subsumes the additive factor of $T(n)/n^2$ from the randomized algorithm assuming $\ell\leq n$).

Setting $\ell=\frac{\log n}{\beta}$ completes the proof.
\end{proof}

\subsection{De-amortizing the number of recolorings}\label{app:deam}
We note that our algorithm can be easily modified to achieve the same trade-off between number of colors and number of recolorings in the worst-case setting as in the amortized setting. This extension does not give a worst-case bound on the running time, only the number of recolorings. Our analysis of the amortized algorithm already uses a trivial black-box dynamic algorithm that performs a constant number of recolorings per update in the worst case. We need to show that the static algorithms can be applied with a worst-case number of recolorings per update. 

The worst-case algorithm works as follows. Since we are not concerned with running time, we run our amortized algorithm in the background (without performing any actual colorings). At the end of each interval $I$, our worst-case algorithm immediately recolors $v_I$ and $u_l$ to the color assigned by $\mathcal{A}_I$. We delay the recoloring of the rest of the vertices in the input of $\mathcal{A}_I$. It is important to recolor $v_I$ and $u_l$ immediately because otherwise the balls and bins game does not apply.

From the proof of correctness of the amortized algorithm (Section~\ref{sec:correct}), if $(u,v)$ is an edge and $u$ and $v$ were last recolored according to two different instances of the static algorithm on different levels, or the same instance, then $u$ and $v$ are assigned different colors. The only remaining case is if $u$ and $v$ were last recolored by different instances of the static algorithm on the same level. The proof that this is impossible from Section~\ref{sec:correct} holds if the following property holds: for every pair of adjacent intervals $I$ and $I'$ on the same level $i$ where $I$ comes before $I'$, all vertices colored by $\mathcal{A}_I$ are recolored by some $\mathcal{A}_{I''}$ on a level $j\not=i$ before any vertices are colored by $\mathcal{A}_{I'}$. 

We design the worst-case algorithm to ensure that this property holds. It suffices to take all of the at most $n/2^i$ vertices input to $\mathcal{A}_I$ and recoloring them to the color assigned by $\mathcal{A}_{I''}$ throughout the course of interval $I'$. For ease of notation, we say that these colorings are \emph{performed by} interval $I'$. We note that by the construction of the intervals,  interval $I''$ ends when interval $I'$ begins so when interval $I'$ begins, it already has full information about all of the recolorings it will perform. Furthermore, when interval $I'$ performs a recoloring according to $\mathcal{A}_{I''}$, the interval following $I''$ on level $j$ has not ended (or even started) yet so these recolorings cannot conflict with other vertices colored using the level $j$ color palette.

 Interval $I'$ is of length $n\ell/2^i$ and performs at most $n/2^i$ recolorings, so on average $I'$ performs at most one recoloring every $\ell$ updates. To achieve $\frac{\log n}{\ell}$ recolorings per update in the worst case, we need to only allow intervals on a $1/\ell$ fraction of the levels to perform recolorings following each update. One way to do this is only allow intervals on level $i$ to perform recolorings after the $k^{th}$ update if $k\equiv i \mod \ell$.

%what about getting even fewer colors like a fraction of Delta?

\section{Algorithms for low arboricity graphs}\label{sec:arbalg}

In this section we begin by proving Theorem~\ref{thm:introarb}. The proof follows from a combination of dynamically maintaining a bounded out-degree edge orientation and applying the algorithm from Section~\ref{sec:gen}. Our main goal in this section is to improve upon Theorem~\ref{thm:introarb} by proving Theorem~\ref{thm:introarbetter}. To this end we refine the tool of dynamic edge orientations by introducing a new \emph{layered data structure}. 
%The bulk of this section is to prove a  Theorem~\ref{thm:intro-lds}, and then we apply it to prove .
% We defer the proof of Theorem~\ref{thm:introarb} to the full version. %Appendix~\ref{app:arbthm}.

\begin{theorem}[Restatement of Theorem~\ref{thm:introarb}]
There is a fully dynamic deterministic algorithm for graphs of arboricity $\alpha$ that maintains an $O((\frac{\alpha}{\beta})^2\log^4 n)$-coloring in amortized $T'(\frac{\alpha\log^3 n}{\beta^2},n)+O(\beta)$ time per update for any $\beta>0$.
Using randomization (against an oblivious adversary), the number of colors can be reduced by a factor of $\hat{\Theta}(\log n)$ and the expected amortized update time becomes
$O(\beta)$. 
\end{theorem}

\begin{proof}
We run the dynamic edge orientation algorithm of~\cite{HTZ14}, which maintains an
$O(\beta'\alpha)$ out-degree orientation in amortized
$O(\frac{\log n}{\beta'})$ time, for all $\beta'>1$. Given this orientation, for any subset $S$ of the vertices in the current graph $G$, we can compute $G[S]$ and a $2\alpha$-coloring of $G[S]$ in time $O(n\beta'\alpha)$. We compute $G[S]$ by simply scanning the out-neighborhood of every vertex in $S$ and including the edges whose other endpoint is also in $S$. Every edge between a pair of vertices $u,v\in S$ is oriented away from either $u$ or $v$ so this algorithm scans every edge in $G[S]$.

Every subgraph of an arboricity $\alpha$ graph also has arboricity $\alpha$, in particular $G[S]$. We color $G[S]$ by considering an ordering of the vertices $v_1,v_2,\dots$ in $S$ such that every vertex has at most $2\alpha$ neighbors that appear after it in the ordering. Such an ordering exists and can be computed in time $O(|S|\alpha)$~\cite{arikati1997efficient}. We imagine starting with an empty graph iteratively adding the vertices in $S$ to the graph in reverse order. When each vertex $v$ is added, $v$ has at most $2\alpha$ neighbors in the current graph. Using a palette of $2\alpha+1$ colors, we can always color $v$ with a color different from all of its neighbors in the current graph.

Applying Theorem~\ref{thm:bodytime} with $T(n)=O(n\beta'\alpha)$ and parameter $\beta''$, we see that the algorithm from Theorem~\ref{thm:body-main} runs in time $O(\beta''\beta'\alpha)$ per update in expectation in the randomized setting and $T'(\frac{\log^2 n}{\beta''},n)+O(\beta''\beta'\alpha)$ per update in the deterministic setting. The additional time for maintaining the edge orientation is  $O(\frac{\log n}{\beta'})$. Setting $\beta'=\sqrt{\log n/(\alpha\beta'')}$, the running time is $O(\sqrt{\alpha\beta''\log n)}$ (with an additional additive factor of $T'(\frac{\log^2 n}{\beta''},n)$ in the deterministic setting).

Applying Theorem~\ref{thm:body-main} with $C=O(\alpha)$, the number of colors is $O(\frac{\alpha}{\beta''}\log^3 n)$ for the deterministic algorithm and $O(\frac{\alpha}{\beta''}\log^2 n(\log\log n+\log \beta''))$ for the randomized algorithm. Setting $\beta''=\beta^2/(\alpha\log n)$ completes the proof.
\end{proof}

For the remainder of this section we prove Theorem~\ref{thm:introarbetter}.

\begin{theorem}[Restatement of Theorem~\ref{thm:introarbetter}]\label{thm:bodyarbetter}
There is a fully dynamic deterministic algorithm for graphs of arboricity $\alpha$ that maintains an $O(\alpha\log^2 n)$-coloring in amortized $\hat{O}(\polylog \alpha)$ time. Using randomization (against an oblivious adversary), the expected amortized time becomes $O(1)$.
\end{theorem}

Given a partition of the vertices of a graph into layers $L_1, L_2,\dots$, for all vertices $v$ let $d_{up}(v)$ (the \emph{up-degree} of $v$) be the number of neighbors of $v$ in layers equal to or higher than that of $v$.

\begin{definition} Given a dynamic graph $G$ of arboricity $\alpha$, a \emph{layered data structure (LDS)} with parameters $k$ and $\Delta$ is a partition of the vertices into $k$ layers $L_1, \dots , L_k$ so that for all vertices $v$, $d_{up}(v)\leq \Delta$.
\end{definition}

The bulk of the proof of Theorem~\ref{thm:bodyarbetter} is to prove Theorem~\ref{thm:intro-lds}.

\begin{theorem}[Restatement of Theorem~\ref{thm:intro-lds}]\label{thm:bodylds}
Let $A''$ be an algorithm for arboricity $\alpha$ graphs that maintains an orientation of the edges with out-degree at most $D$ that performs amortized $F(n)$ flips per update. Then there is an algorithm to maintain an LDS along with the graph induced by each layer, for a fully dynamic graph of arboricity $\alpha$ with $k=O(\log n)$ and $\Delta=O(D+\alpha\log n)$ in amortized deterministic time $O(F(n))$.
\end{theorem}

We note that we do not require the algorithm $A''$ to be explicit; we only require its existence.

\subsection{Proof overview}\label{sec:ldsover}

The idea of the algorithm is essentially to move vertices to new layers when the required properties of the data structure are violated. Roughly, when there is a vertex $v$ with $d_{up}(v)\geq\Delta$ we move $v$ to a higher layer so that $d_{up}(v)$ decreases to $O(\alpha)$. Additionally, to control the number of layers, whenever a vertex $v$ has up-degree less than $d=O(\alpha)$ and $v$ can be moved to a lower layer while maintaining up-degree less than $d$, we move $v$ to a lower layer. The fact that $d$ and $\Delta$ differ by a logarithmic factor ensures that vertices don't move between layers too often which is essential for bounding the running time.

To help with the running time analysis, we maintain \emph{two} dynamic orientations of the edges: one is defined by the algorithm $A''$ and the other is maintained by our algorithm. The orientation maintained by our algorithm has the property that all edges with endpoints in different layers are oriented toward the higher layer. We compare the number of edge flips in the orientation defined by our algorithm to the number of edge flips in the orientation algorithm $A''$ using a potential function: $\phi(i)=$ the number of edges oriented in opposite directions in the two algorithms. This potential function is also used in~\cite{BF99}. 

The main idea of the analysis is to observe how $\phi$ changes in response to vertices moving between levels. We claim that when we move a vertex to a higher level, $\phi$ decreases substantially for the following reason. Our algorithm is defined so that we only move a vertex to a higher layer if its up-degree decreases substantially as a result. Because our algorithm orients edges from lower to higher layers, when we move a vertex $v$ to a higher layer many edges incident to $v$ are flipped towards $v$. Then because $A''$ maintains an orientation of low out-degree, many of these edges flipped towards $v$ end up oriented in the same direction in the two orientations. Thus, $\phi$ decreases substantially as a result of $v$ moving to a higher layer. On the other hand, when a vertex moves to a lower layer, $\phi$ might increase. The idea of the argument is to use the substantial decreases in $\phi$ that result from moving vertices to higher layers to pay for the increases in $\phi$ that result from moving vertices to lower layers. 

%todo: note to self add yuan paper to citations for practical dynamic coloring

\subsection{Invariants}

In this section we introduce four invariants that together imply that $d_{up}(v)\leq \Delta$ and $k=O(\log n)$.

We maintain two dynamic orientations of the edges in the graph, one defined by our algorithm and the other defined by the algorithm $A''$. Unless otherwise stated, when we refer to an orientation, we mean the orientation defined by our algorithm.

For ease of notation, let $d=4\alpha$ and let $d'=\Delta/2$.

\noindent We define the following for each vertex $v$:
\begin{itemize}
\item $L(v)$ is the layer containing $v$.
\item $L_{max}(v)$ is the lowest layer for which if $v$ were in this layer, $d_{up}(v)$ would be at most $d$.
\item $d^+(v)$ is the out-degree of $v$. 
\item $d_L^-(v)$ is the in-degree of $v$ from neighbors in $L(v)$.
\end{itemize}

\subsubsection{Orientation invariants}

Invariant~\ref{inv:up} defines how edges are oriented between layers and is useful for analyzing the update time of the algorithm, as outlined in Section~\ref{sec:ldsover}. 

\begin{invariant} \label{inv:up} All edges with endpoints in different layers are oriented towards the vertex in the higher layer.\end{invariant}

%todo include preliminaries section to define notation of induced subgraph, d+(v) etc.? no, just explain in context.

The next two invariants bound $d^+(v)$ and $d_L^-(v)$, which helps to bound $d_{up}(v)$.

\begin{invariant} \label{inv:out} For all vertices $v$, $d^+(v)\leq d'$. \end{invariant}
\begin{invariant} \label{inv:in} For all vertices $v$, $d_L^-(v)\leq d'$. \end{invariant}

\begin{claim}\label{claim:inv1}
Invariants~\ref{inv:up}-\ref{inv:in} together imply that $d_{up}(v)\leq 2d'=\Delta$.
\end{claim}

\begin{proof}
By Invariant~\ref{inv:up}, for all vertices $v$, every neighbor of $v$ in a layer equal to or higher than $L(v)$ is either an out-neighbor of $v$ or an in-neighbor of $v$ in $L(v)$, so $d_{up}(v)=d^+(v)+d_L^-(v)$. Then by Invariants~\ref{inv:out} and \ref{inv:in}, $d^+(v)+d_L^-(v)\leq 2d'$.
\end{proof}

\subsubsection{Number of layers invariant}

Invariant~\ref{inv:max} serves to bound the number of layers $k$. 

For any pair of layers $L_i$, $L_j$, we abuse notation and say that $L_i<L_j$ if $i<j$, that is if layer $L_i$ is below layer $L_j$.

\begin{invariant} \label{inv:max} For all vertices $v$, $L(v)\leq L_{max}(v)$. \end{invariant}

\begin{claim}\label{claim:inv2} Invariant~\ref{inv:max} implies that $k = O(\log n)$.\end{claim}
\begin{proof}
First we observe that under Invariant~\ref{inv:max}, all vertices of degree at most $d$ are in $L_1$. Now, consider removing all vertices in $L_1$ from the graph. In the remaining graph, all vertices of degree at most $d$ are in layer $L_2$. More generally, after removing all vertices in layers 1 through $i$ for any $i$, all vertices of degree at most $d$ must be in layer $L_{i+1}$. 

The total number of edges in a graph of arboricity alpha is less than $\alpha n$. So at least a $(1-2\alpha/d)$ fraction of the vertices have degree at most $d$. Any subgraph of an arboricity $\alpha$ graph also has arboricity $\alpha$ so after the vertices in any given layer are removed, the graph still has arboricity $\alpha$. Thus, after removing the vertices in layers 1 through $i$ for any $i$, at least a $(1-2\alpha/d)$ fraction of the remaining vertices are in $L_{i+1}$. Therefore, the number $k$ of layers total is at most $\log_{\frac{d}{2\alpha}} n=O(\log n)$.
\end{proof}

\subsection{Algorithm}\label{sec:alg}

The idea of the algorithm is essentially to move vertices to new layers when the required properties of the data structure are violated. We define two recursive procedures \textsc{Rise} and \textsc{Drop} which move vertices to higher and lower layers respectively. In particular, when a vertex $v$ violates Invariant~\ref{inv:out} or \ref{inv:in} (i.e. either $d^+(v)>d'$ or $d_L^-(v)>d'$), we call the procedure \textsc{Rise}$(v)$ which moves $v$ up to the layer $L_{max}(v)$. The movement of $v$ to a new higher layer may increase the up-degree of some neighbors $u$ of $v$ causing $u$ to violate Invariant~\ref{inv:out} or \ref{inv:in}, in which case we recursively call \textsc{Rise}$(u)$. On the other hand, when a vertex $v$ violates Invariant~\ref{inv:max} (i.e. $L_{max}(v)<L(v)$), we call the procedure \textsc{Drop}$(v)$ which moves $v$ down to the layer $L_{max}(v)$. The movement of $v$ to a new lower layer may decrease $L_{max}(u)$ for some neighbors $u$ of $v$ causing $u$ to violate Invariant~\ref{inv:max}, in which case we recursively call \textsc{Drop}$(u)$. See Algorithm~\ref{alg} for the pseudocode. 
\begin{algorithm}[h]
\caption{}\label{alg}
\begin{algorithmic}

\Procedure {\textsc{Insert}($u$,$v$)}{}
\State add edge $(u,v)$
\If {$u$ and $v$ are in different layers}
\State orient the edge towards the vertex in the higher layer
\Else ($u$ and $v$ are in the same layer):
\State orient the edge arbitrarily
\EndIf
\If {$d^+(u)>d'$ or $d_L^-(u)>d'$}
\Call{\textsc{Rise}}{$u$}
\EndIf
\If {$d^+(v)>d'$ or $d_L^-(v)>d'$}
\Call{\textsc{Rise}}{$u$}
\EndIf
\EndProcedure

\Procedure {\textsc{Delete}($u$,$v$)}{}
\State remove edge $(u,v)$
\If {$L_{max}(u)<L(u)$}
	\Call{\textsc{Drop}}{$u$}
	\EndIf
\If {$L_{max}(v)<L(v)$}
	\Call{\textsc{Drop}}{$v$}
	\EndIf
	\EndProcedure

\Procedure{\textsc{Rise}($v$)}{}
\State $L_{old}\gets L(v)$
\State move $v$ up to layer $L_{max}(v)$
\State $S\gets $ the set of out-neighbors of $v$ in a layer between $L_{old}$ and $L_{max}(v)$ inclusive
\For {each $u \in S$}
	\State flip edge $(u,v)$ towards $v$
	\EndFor
\For {each $u \in S$}{}
	\If {$d^+(u)>d'$}
		\Call{\textsc{Rise}}{$u$}
		\EndIf
		\EndFor
			\EndProcedure

\Procedure{\textsc{Drop}($v$)}{}
\State $L_{old}\gets L(v)$
\State move $v$ down to layer $L_{max}(v)$
\State $S \gets$ the set of in-neighbors of $v$ in any layer above $L_{max}(v)$ and at most $L_{old}$
\For {each $u \in S$}
	\State flip edge $(u,v)$ away from $v$
	\EndFor
\State $S^+ \gets$ the set of all neighbors of $v$ in any layer above $L_{max}(v)$ and at most $L_{old}+1$
\For {each $u \in S^+$}{}
	\If {$L_{max}(u)<L(u)$}
		\Call{\textsc{Drop}}{$u$}
		\EndIf
		\EndFor
			\EndProcedure
\end{algorithmic}
\end{algorithm}

\subsection{Correctness}

We will show that after each edge update is processed, the four invariants are satisfied. (The edge update algorithm indeed terminates due to the running time analysis in the following sections.) By Claims~\ref{claim:inv1} and~\ref{claim:inv2}, this implies that $d_{up}(v)\leq 2d'=\Delta$ and $k = O(\log n)$.

We will use the following useful property of the algorithm:

\begin{lemma} \label{lem:ins}$ $
\begin{enumerate}%[label=\alph]
\item Right after any vertex $v$ is moved to a new layer $L_i$, $d_{up}(v)\leq d$.
\item While $v$ remains in $L_i$, the only way for $d_L^-(v)$ to increase is by the insertion of an edge incident to $v$.
\end{enumerate}
\end{lemma}

\begin{proof}$ $
\begin{enumerate}%[a)]
\item Whenever any vertex $v$ is moved to a new layer (either by $\textsc{Rise}$ or \textsc{Drop}), it is moved to the layer $L_{max}(v)$. By the definition of $L_{max}(v)$, we have $d_{up}(v)\leq d$.
\item When any vertex $v$ moves to a higher layer, all of $v$'s incident edges within its new layer are flipped towards $v$. When any vertex $v$ moves to a lower layer, all of $v$'s incident edges within its new layer are already oriented towards $v$ by Invariant~\ref{inv:up} and they are not flipped. That is, right after $v$ is moved to a new layer (in either direction), all of its incident edges within its new layer are oriented towards $v$. Thus, for all vertices $u\not=v$, the movement of $v$ to a new layer cannot cause $d_L^-(u)$ to increase. Then since all edge flips are triggered by a vertex changing layers, the only way for $d_L^-(v)$ to increase is by the insertion of an edge incident to $v$.
\end{enumerate}
\end{proof}

Now we show that the four invariants are satisfied after each edge update is processed.

Invariant~\ref{inv:up} is satisfied at all times because whenever a vertex changes layer all of its incident edges that are oriented towards the lower layer are immediately flipped.

Invariant~\ref{inv:out} is violated when $d^+(v)> d'$. This could happen as a result of a) insertion of an edge, or b) movement of a vertex $u$ to a higher layer, which could cause $u$'s neighbors to violate the invariant. In both of these cases, the algorithm calls \textsc{Rise} on all violating vertices.

Invariant~\ref{inv:in} is violated when $d_L^-(v)> d'$. By Lemma~\ref{lem:ins}, this can only happen following the insertion of an edge. In this case, the algorithm calls \textsc{Rise} on all violating vertices.

Invariant~\ref{inv:max} is violated when $L(v)> L_{max}(v)$. This could happen as a result of a) deletion of an edge, or b) movement of a vertex $u$ from $L_i$ to a lower layer $L_j$, which could cause $u$'s neighbors in layers from $L_{j+1}$ to $L_{i+1}$ to violate the invariant. In both of these cases, the algorithm calls \textsc{Drop} on all violating vertices.

\subsection{Bounding the number of edge flips}

The first step towards getting a bound on the update time is to get a bound on the number of edge flips that the algorithm performs. We will show that the amortized number of edge flips per update is $O(F(n))$ (Lemma~\ref{lem:amflips}).

We choose $\Delta=16(dk+D)$, so $\Delta=O(D+\alpha \log n)$, as required.

Let ${\mathcal G} = (G_0,G_1,\dots,G_M)$ be the sequence of graphs with orientation defined by our algorithm and let ${\mathcal G}^A = (G^A_0,G^A_1,\dots,G^A_M)$ be the sequence of graphs with orientation defined by the algorithm $A''$. That is, for all $i$, the underlying undirected graphs corresponding to $G_i$ and $G^A_i$ are identical but their orientations may differ. Given i, we say an edge in $G_i$ is \emph{bad} if it is oriented in the opposite direction in $G_i$ and $G^A_i$. We define a potential function: \[\phi(i) = \mbox{the number of bad edges.}\]
We say that a call to \textsc{Rise} is \emph{heavy} if it triggers at least $d'/2$ edge flips, ignoring recursive calls. Otherwise, we say that a call to \textsc{Rise} is \emph{light}.

\begin{lemma}\label{lem:light} Every light call to \textsc{Rise} is due to a violation of Invariant~\ref{inv:in}. 
%($d_L^-(v)\leq d'$).
\end{lemma}

\begin{proof}
Suppose otherwise; that is, suppose that a light call to \textsc{Rise} is triggered by a violation of Invariant~\ref{inv:out}. In this case, right before the call to \textsc{Rise}, $d^+(v)>d'$. By Lemma~\ref{lem:ins}, after $v$ is moved to a new layer, $d^+(v)\leq d$. Thus, the call to \textsc{Rise} triggers at least $d'-d>d'/2$ edge flips so it must be heavy.
\end{proof}

We define the following parameters. \\
$M =$ the total number of edge updates\\
$l =$ the total number of light calls to \textsc{Rise}\\
$h =$ the total number of heavy calls to \textsc{Rise}\\
$r =$ the total number of calls to \textsc{Rise}; so $r=h+l$\\
$p =$ the total number of levels that vertices move down (due to calls to \textsc{Drop})\\
$f$ = the total number of flips\\
$f_l =$ the total number of flips triggered by light calls to \textsc{Rise}\\
$f_h =$ the total number of flips triggered by heavy calls to \textsc{Rise}\\
$f_p$ = the total number of flips triggered by calls to \textsc{Drop}\\

\begin{observation}\label{obs:phi}
Using the above parameters, it is immediate to bound the total increase in $\phi$ due to the following events:
\begin{itemize}
\item Edge updates: $\Delta(\phi)\leq M$.
\item Edge reorientations in $G^A$: $\Delta(\phi)\leq F(n)M$.
\item Light calls to \textsc{Rise}: $\Delta(\phi)\leq f_l$.
\item Calls to \textsc{Drop}: $\Delta(\phi)\leq f_p$.
\end{itemize}
\end{observation}

The only event missing from the above list is heavy calls to \textsc{Rise}. We will now argue that $\phi$ decreases substantially as a result of this event. Then, we will use these substantial decreases in $\phi$ to pay for the increases in $\phi$ from the other events.

\begin{lemma}\label{lem:heavy} The total decrease in $\phi$ over the whole computation triggered by heavy calls to \textsc{Rise} is at least $f_h/2$. \end{lemma}

\begin{proof}
Consider a heavy call to \textsc{Rise} on vertex $v$. Let $S$ be the set of edges flipped by this call to \textsc{Rise}. All of the edges in $S$ are flipped towards $v$. Before these flips happen, $v$ has out-degree at least $d'/2$ by the definition of a heavy call to \textsc{Rise}. By Lemma~\ref{lem:ins}, after the edges in $S$ are flipped, $d^+(v)\leq d$. Thus, the number of edges flipped is at least $d'/2-d$. We will use this fact at the end of the proof.

	Before the edges in $S$ are flipped, all are out-going of $v$. Then since $v$ has out-degree at most $D$ in all $G^A_i$, at least $|S|-D$ of these edges are bad before they are flipped. For the same reason, after these flips at most $D$ of the flipped edges are bad. Thus, $\phi$ decreases by at least $|S|-2D$ as a result of flipping the edges in $S$. Therefore, the total decrease in $\phi$ over the whole computation due to heavy calls to \textsc{Rise} is at least the sum of $|S|-2D$ over all heavy calls to \textsc{Rise}, which is at least
	\begin{align*} f_h-2Dh&\geq  f_h-4Df_h/d' \text{ since each heavy call to \textsc{Rise} flips at least $d'/2$ edges}\\&\geq \frac{f_h}{2} \text{ by choice of $d'$}.
\end{align*}
\end{proof}

We derive bounds for $f_l$, $f_p$, and $f_h$, in the following lemmas.

\begin{lemma}\label{lem:l} $l(d'-d)\leq M$. $f_l\leq M$. \end{lemma}

\begin{proof}	
By Lemma~\ref{lem:light} every light call to \textsc{Rise} is triggered by $d_L^-(v)>d$ for some $v$. By Lemma~\ref{lem:ins}, right after any vertex $v$ is moved to a new layer, $d_L^-(v)\leq d$, and while $v$ remains in this layer, the only way for $d_L^-(v)$ to increase is by the insertion of an edge incident to $v$. Before a light call to \textsc{Rise}($v$), $d_L^-(v)$ must increase to at least $d'$. Thus, every light call to \textsc{Rise} must be preceded by $d'-d$ insertions of edges incident to $v$. Conversely, the insertion of an edge can only increase the in-degree of one vertex. Thus, $l(d'-d)\leq M$. Each light call to \textsc{Rise} flips at most $d'/2$ edges so, $f_l\leq \frac{ld'}{2}$. Combining these two equations we have, $f_l\leq \frac{Md'}{2(d'-d)}\leq M$ by choice of $d'$.
\end{proof}

\begin{lemma}\label{lem:p} $f_p\leq dp\leq \frac{f_h}{4}+M$.\end{lemma}

\begin{proof} By Lemma~\ref{lem:ins}, right after the call to \textsc{Drop}($v$), $d_{up}(v)\leq d$. Then since \textsc{Drop}($v$) only flips edges incident to $v$ whose other endpoint is in a layer above $v$, any call to \textsc{Drop}($v$) flips at most $d$ edges. Thus, $f_p\leq dp$. Furthermore, every call to \textsc{Rise}($v$) moves $v$ up by at most $k$ layers, so $p\leq rk$. Thus we have,
\begin{align*}
f_p&\leq drk \\
&= dk(h+l)\\
&\leq dk\left(\frac{2f_h}{d'}+l\right) \text{ since each heavy call to \textsc{Rise} flips at least $d'/2$ edges}\\
&\leq dk\left(\frac{2f_h}{d'}+\frac{M}{d'-d}\right) \text{ by Lemma~\ref{lem:l}}\\
&\leq f_h/4+M \text{ by choice of $d'$}
\end{align*}
\end{proof}

\begin{lemma}
$f_h=O(F(n)M)$
\end{lemma}
\begin{proof}
We use the potential function: $\phi$ is initially 0 and is never negative so the total increase in $\phi$ must be at least the total decrease in $\phi$. Therefore, by Observation~\ref{obs:phi} and Lemma~\ref{lem:heavy},
$M+F(n)M+f_l+f_p\geq f_h/2$.
Then, by Lemmas~\ref{lem:l} and \ref{lem:p}, we have $M+F(n)M+M+f_h/4+M\geq f_h/2$, which completes the proof.
\end{proof}

\begin{lemma}\label{lem:amflips}
The amortized number of flips per update is $O(F(n))$.
\end{lemma}
\begin{proof}
\begin{align*}
f&= f_p+f_l+f_h\\&\leq f_h/4+M+M+f_h \text{ by Lemmas~\ref{lem:l} and~\ref{lem:p}}\\
&=O(F(n)M) \text{ by Lemma~\ref{lem:heavy}}
\end{align*}
\end{proof}

\subsection{Update time bound}

In this section, we will show that our algorithm runs in amortized time $O(F(n))$ per update.

Each vertex $v$ keeps track of the following information:
\begin{itemize}
\item $L(v)$
\item $d^+(v)$ and the set $N^+(v)$ of $v$'s out-neighbors
\item $d_L^-(v)$ and the set $N_L^-(v)$ of $v$'s in-neighbors in $L(v)$.
\item For each layer $L_i$ lower than $L(v)$, the set $N_i(v)$ of $v$'s neighbors in that layer and the number $d_i(v)$ of them.
\end{itemize}

%We will show that each of the four procedures in our algorithm take amortized time $O(T''(n))$, completing the proof of Theorem~\ref{thm:lds} todo correct ref?. First we show that for the \textsc{Rise} and \textsc{Drop} procedures, it suffices to show that each call to \textsc{Rise} takes time $O(d_{up}(v))$ (with respect to $v$'s layer before the call) and each call to \textsc{Drop} takes time $O(d)$.

We require that insertion and deletion of elements to and from the subsets of vertices that we maintain both take constant time. This is possible, for example, by using an array of length $n$ and flipping the bit corresponding to the inserted or deleted vertex.

\begin{lemma}
\textsc{Insert}(u,v) runs in time $O(1)$ (ignoring calls to \textsc{Rise}).
\end{lemma}
\begin{proof}
In \textsc{Insert}(u,v) we update the stored information of both $u$ and $v$ in constant time simply by incrementing the appropriate counters and adding to the appropriate sets. Then we compute whether either $d^+(u)$ or $d_L^-(u)$ exceeds $d'$ and the same for $v$. These comparisons take $O(1)$ time.
\end{proof}

\begin{lemma}\label{lem:lmax} Computing whether $L_{max}(v)<L(v)$ takes time $O(1)$.\end{lemma}

\begin{proof}
$L_{max}(v)<L(v)$ if and only if the degree of $v$ to vertices in layers at least as high as the layer just below $L(v)$ is at most $d$ i.e. if $d^+(v)+d_L^-(v)+d_{i-1}(v)\leq d$ where $i$ is such that $L_i=L(v)$. This comparison takes $O(1)$ time.
\end{proof}

\begin{lemma}
\textsc{Delete}(u,v) runs in time $O(1)$ (ignoring calls to \textsc{Drop}).
\end{lemma}

\begin{proof}\textsc{Delete($u,v$)} updates the stored information of both $u$ and $v$ in constant time simply by decrementing the appropriate counters and deleting from the appropriate sets. Then \textsc{Delete(u,v)} computes whether $L_{max}(v)<L(v)$ and whether $L_{max}(u)<L(u)$. This takes $O(1)$ time by Lemma~\ref{lem:lmax}.\end{proof}

\begin{lemma}\label{lem:rise}
Ignoring recursive calls, \textsc{Rise}($v$) runs in time $d_{up}(v)$ with respect to $v$'s layer immediately before the call to \textsc{Rise}($v$).
\end{lemma}

\begin{proof}
It takes time $O(d_{up}(v))$ to scan the set $N^+(v)\cup N_L^-(v)$, which suffices to determine $L_{max}(v)$, build the set $S$ (defined in Algorithm~\ref{alg}), flip the appropriate edges, and determine for each $u\in S$ whether $d^+(u)>d$.

We must also update the stored information for $v$ and all vertices in $N^+(v)\cup N_L^-(v)$. This can be done in $O(d_{up}(v))$ time by incrementing/decrementing the appropriate counters and editing the appropriate sets. Importantly, every vertex in a layer below $v$ immediately before the call to \textsc{Rise}($v$) does not need to update its stored information because vertices only keep track of the exact layer of their neighbors on lower layers. Additionally, $v$ does not need to update any of its information concerning its neighbors in lower layers.

Additionally, when $v$ changes layer we update the graph induced by its old and new layers. All of $v$'s incident edges to vertices in its old layer are removed from the graph induced by its old layer and all of $v$'s incident edges to vertices in its new layer are added to the graph induced by its new layer. There are at most $d_{up}(v)$ edges (with respect to $v$'s layer before being moved).
%\begin{itemize}
%\item Updating $L(v)$ takes constant time
%\item Updating $d^+(v)$ and the list of $v$'s out-neighbors takes $O(d^+(v))$ time.
%\item $d_L^-(v)$ is set to 0.
%\item Updating for each layer at most $L(v)$, a list of $v$'s neighbors in that layer and the number $d_i(v)$ of them takes $O(d_{up}(v))$ time since the only vertices to be added to or removed from these sets are those initially in higher layers than $v$.
%\item For each flipped edge $(u,v)$, $v$ must be added to the list of out-neighbors of $u$ and $d^+(v)$ must be incremented. This takes time proportional to the number of edges flipped, which is at most $d_{up}(v)$.
%\item For each up-neighbor $u$ of $v$, $v$ needs to be removed from the appropriate list $S_i(u)$ and possibly added to the appropriate list $S_i(v)$. This takes time $O(d_{up}(v))$.
%Lastly, \textsc{Rise}(v) determines whether $d^+(u)$ exceeds $d'$, which takes constant time.
\end{proof}

\begin{lemma}\label{lem:risetime} \textsc{Rise}($v$) runs in amortized $O(F(n))$ time.\end{lemma}

\begin{proof}
By Lemma~\ref{lem:rise}, if we ignore recursive calls, each call to \textsc{Rise}($v$) takes time $O(d_{up}(v))$ with respect to $v$'s layer immediately before the call to \textsc{Rise}($v$).

We claim that when \textsc{Rise}($v$) is called, $d_L^-(v)\leq d'+1$. By Lemma~\ref{lem:ins}, the only operation that can trigger $d_L^-(v)$ to exceed $d'$ is an edge insertion. If such an edge insertion happens, \textsc{Rise}($v$) is immediately called at which point $d_L^-(v)=d'+1$.

All terms in the following inequalities are with respect to $v$'s layer immediately before the call to \textsc{Rise}($v$). Let be $g$ be the number of edges flipped in the call to \textsc{Rise}($v$) (ignoring recursive calls). Then $g\geq d^+(v)-d$ since the out-degree of $v$ after the edge flips is at most $d$.
\begin{align*}
d_{up}(v)&=d^+(v)+d_L^-(v)\\
&\leq d^+(v)+d'+1\text{\hspace{2mm} since $d_L^-(v)\leq d'+1$}\\
&\leq g+d+d'+1\text{\hspace{2mm} since $g\geq d^+(v)-d$.}
\end{align*}
Thus, if \textsc{Rise}($v$) is a heavy call then $d_{up}(v)=O(g)$. By Lemma~\ref{lem:amflips}, the amortized number of flips per update is $O(F(n))$ so heavy calls to \textsc{Rise} run in amortized time $O(F(n))$. On the other hand, if \textsc{Rise}($v$) is a light call, then $d_{up}(v)=O(d')$. Thus, the total time for light calls to \textsc{Rise} is
\begin{align*}
O(ld')&=O(\frac{Md}{d'-d})\text{ by Lemma~\ref{lem:l}}\\
& = O(M) \text{ by choice of $d'$}
\end{align*} so the amortized time for light calls to \textsc{Rise} is $O(1)$.
\end{proof}

\begin{lemma}\label{lem:droptime} \textsc{Drop} runs in amortized $O(F(n))$ time.\end{lemma}

\begin{proof} \textsc{Drop}($v$) begins by computing $L_{max}(v)$. Let $j$ be such that $L_j=L(v)$. $L_{max}(v)$ can be caluculated by finding the value $t$ such that $d_{up}(v)+\sum_{i=1}^t d_{{j-i}}(v)\leq d$ but $d_{up}(v)+\sum_{i=1}^{t+1} d_{{j-i}}(v)> d$.  The time spent doing this calculation is proportional to the number of layers that $v$ moves in this call to \textsc{Drop}($v$). Thus, the time spent on these calculations throughout the whole computation is $O(p)=O(f_h/4+M)$ by Lemma~\ref{lem:p}, which is $O(F(n)M)$ by Lemma~\ref{lem:heavy}. 

After moving $v$ to layer $L_{max}(v)$, \textsc{Drop}($v$) builds the sets $S$ and $S^+$ (defined in Algorithm~\ref{alg}) and flips the appropriate edges. To build the sets $S$ and $S^+$, we scan $N^+(v)$ and $N_i(v)$ for the appropriate layers $L_i$. This takes constant work for each of the $O(d)$ elements in $S$ and $S^+$ plus constant work for each $N_i(v)$ scanned. The number of $N_i(v)$ scanned is the number of layers that $v$ moves in this call to \textsc{Drop}($v$). From the calculations in the previous paragraph, scanning these $N_i(v)$ takes total time $O(F(n)M)$ over the whole computation. Additionally by Lemma~\ref{lem:lmax}, determining whether $L_{max}(u)<L(u)$ for each $u\in S^+$ takes $O(d)$ time (constant time for each vertex $u$).

We must also update the stored information for $v$ and all vertices in $N^+(v)\cup N_L^-(v)$ (with respect to $v$'s new layer). This can be done in $O(d)$ time by incrementing/decrementing the appropriate counters and editing the appropriate sets. Importantly, every vertex in a layer below $v$'s new layer does not need to update its stored information because vertices only keep track of the exact layer of their neighbors on lower layers. Additionally, $v$ does not need to update any of its information concerning its neighbors in lower layers.

Additionally, when $v$ changes layer we update the graph induced by its old and new layers. All of $v$'s incident edges to vertices in its old layer are removed from the graph induced by its old layer and all of $v$'s incident edges to vertices in its new layer are added to the graph induced by its new layer. By construction $v$ has at most $d$ edges incident to vertices in its old layer, and by Lemma~\ref{lem:ins}, $v$ has at most $d$ edges incident to its new layer. 

%At this point, we must update the following quantities:
%\begin{itemize}
%\item Updating $L(v)$ takes constant time
%\item Updating $d^+(v)$ and the list of $v$'s out-neighbors takes $O(d)$ time.
%\item $d_L^-(v)$ is set to 0.
%\item Updating for each layer at most $L(v)$, the set of $v$'s neighbors in that layer and the number $d_i(v)$ of them takes $O(d)$ time since the only vertices to be added to or removed from these sets are those initially in layers equal to or above $L(v)$.
%\item For each flipped edge $(u,v)$, $v$ must be added to the set of out-neighbors of $u$ and $d^+(v)$ must be incremented. This takes time proportional to the number of edges flipped, which is at most $d_{up}(v)$.
%\item For each up-neighbor $u$ of $v$, $v$ needs to be removed from the appropriate list $S_i(u)$ and possibly added to the appropriate list $S_i(v)$. This takes time $O(d_{up}(v))$.
%Lastly, for a subset of the vertices $u$ in $N_{up}(v)$, \textsc{Drop}(u,v) determines whether $L_{max}(v)<L(v)$. This takes a total of $O(d)$ time by Lemma todo.

We have shown that running \textsc{Drop} consists of operations that take total time $O(F(n)M)$ throughout the whole computation, plus operations that take time $O(d)$ for each call to \textsc{Drop}. Thus, the total time for calls to \textsc{Drop} is $O(F(n)M+pd)$. By Lemma~\ref{lem:p}, $pd\leq \frac{f_h}{4}+M=O(F(n)M)$ so the overall amortized time is $O(F(n))$.
\end{proof}

\subsection{Coloring from LDS}

\begin{proof}[Proof of Theorem~\ref{thm:bodyarbetter}] Recall that $A'$ is a fully dynamic algorithm that colors graphs of maximum degree $\Delta$ using $O(\Delta)$ colors. Further recall that such a randomized algorithm exists with $O(1)$ amortized update time in expectation and that $T'(\Delta,n)\leq \polylog(\Delta)$ is the running time of an optimal deterministic algorithm for this problem.

By definition, the graph induced by each layer has degree at most $2d'$. We assign each of the $O(\log n)$ layers a disjoint set of $O(d')$ colors and use algorithm $A'$ to dynamically color the graph induced by each layer independently. The total number of colors is $O(d'\log n)=O((D+\alpha\log n)\log n)$.

Since we are maintaining the graph induced by each layer in amortized time $O(F(n))$, the amortized number of edges inserted into or deleted from the graphs induced by each layer is $O(F(n))$. Thus, the amortized update time is $O(F(n)\cdot T'(2d',n))$.

Applying the dynamic edge orientation result of~\cite{kowalik2007adjacency}, which gives $D=O(\alpha\log n)$ and $F(n)=O(1) $ completes the proof.
\end{proof}

\subparagraph*{Acknowledgements} The authors thank Krzysztof Onak, Baruch Schieber and Virginia Vassilevska Williams for discussions.
\bibliography{bibfileFromSTOC18}
\appendix
%\input{recourse}
%\section{Algorithm for general graphs}
\section{Balls and bins bound on recent degree}\label{app:bins}
%TODO add a sentence to make this section also work for the randomized game.
We prove the portion omitted from the proof of Lemma~\ref{lem:deg}. Lemma~\ref{lem:app} is easy to see based on the description of the algorithm and the balls and bins game; we include the rigorous proof for the purpose of completeness.

\begin{lemma}\label{lem:app} In the special case of the balls are bins game defined in the proof of lemma~\ref{lem:deg}, the size of the largest bin is at least the maximum recent degree in the graph.
\end{lemma}

\begin{proof}
Let $B=b_1,\dots,b_n$ be the set of bins in the game and let $B'=b'_1,\dots,b'_n$ be the set of hypothetical bins such that at all times the number of balls in bin $b'_i$ is exactly the recent degree of vertex $i$. For $B'$ we say that Player $1'$ adds balls to $B'$ and Player $2'$ removes them. We will show that at all times the size of the largest bin in $B$ is at least the size of the largest bin in $B'$. We cannot simply say that for all $i$, $|b_i|\geq |b'_i|$ because Player 2 might empty different bins in $B$ than Player $2'$ empties in $B'$. Instead, we will show that there is a permutation $\sigma$ of $[n]$ so that each bin in $B'$ is at least as large as its corresponding bin in $B$. We say that $B$ \emph{dominates} $B'$ (with respect to $\sigma$) if for all $i$, $b_{\sigma(i)}$ contains at least as many balls than $b'_i$. We will show that $B$ indeed dominates $B'$ at all times, which completes the proof.

Initially $\sigma$ is the identity permutation and all bins are empty so $B$ trivially dominates $B'$. Suppose inductively that $B$ dominates $B'$. Only the following two types of operations could cause $B$ to stop dominating $B'$: a) balls are added to bins in $B'$, and b) balls are deleted from bins in $B$.

The first type of operation is not an issue because by definition when a ball is added to a bin $b'_i$ in $B'$ a ball is also added to bin $b_{\sigma(i)}$ in $B$. The second type of operation occurs in the deterministic game when Player 2 empties the largest bin of $B$. When this happens, Player $2'$ empties the largest bin of $B'$. Let $i$ be such that $b_{\sigma(i)}$ is the bin in $B$ that is emptied and let $j$ be such that $b'_j$ is the bin in $B'$ that is emptied. By the inductive hypothesis, $|b_{\sigma(j)}|\geq |b'_j|$ and by choice of $j$, $|b'_j|\geq|b'_i|$. Thus, $|b_{\sigma(j)}|\geq |b'_i|$.  Then, since $|b_{\sigma(i)}|\geq |b'_j|$ (both are set to 0), we can modify $\sigma$ by switching $\sigma(i)$ and $\sigma(j)$, so that $B$ still dominates $B'$.

The second type of operation also occurs in the randomized game when Player 2 empties a random bin as defined in the game. When this happens, Player $2'$ also empties a random bin as defined in the algorithm. By these definitions, Players 2 and $2'$ choose a bin randomly over the same distribution of bins (i.e. if Player 2 chooses bin $b_{\sigma(i)}$ with probability $p$, then Player $2'$ chooses bin $b'_i$ with probability $p$). We can assume that Players 2 and $2'$ share a random coin and select matching bins.
\end{proof}

%\input{app_arb}

% todo: delete vspace for full version

%%
%% Bibliography
%%

%% Please use bibtex,

\end{document}